\documentclass{article}

\usepackage{cite}
\usepackage{fullpage}
\usepackage{definitions}
\usepackage{url}
\usepackage{epsfig}
\usepackage[dvips]{color}
\usepackage{multirow}

\newcommand{\V}{\ensuremath{V}\xspace}
\newcommand{\PM}{\ensuremath{P}\xspace}
\newcommand{\p}[2][]{\ensuremath{p_{#1,#2}}\xspace}
\newcommand{\de}[1][]{\ensuremath{d_{#1}}\xspace}
\newcommand{\cost}[1][]{\ensuremath{c_{#1}}\xspace}

\newcommand{\pay}[1][]{\ensuremath{\paycoef[#1] \maxpay[#1]}\xspace}
\newcommand{\paycoef}[1][]{\ensuremath{\lambda_{#1}}\xspace}

\newcommand{\maxpay}[1][]{\ensuremath{\pi_{#1}}\xspace}
\newcommand{\maxpayv}{\ensuremath{\mathbf{\pi}}\xspace}
\newcommand{\xv}{\ensuremath{\mathbf{x}}\xspace}

\newcommand{\INCLPROB}[1][]{\ensuremath{\ifthenelse{\equal{#1}{}}{q}{q_{#1}}}\xspace}
\newcommand{\InclProb}[2]{\ensuremath{q_{#1}(#2)}\xspace}
\newcommand{\LOCALACT}[1]{\ensuremath{f^{#1}}\xspace}
\newcommand{\localact}[2]{\ensuremath{f^{#1}(#2)}\xspace}
\newcommand{\localc}[2]{\ensuremath{\hat{f}^{#1}(#2)}\xspace}

\newcommand{\der}[3]{\ensuremath{#1_{#2}(#3)}\xspace}
\newcommand{\DISTR}{\ensuremath{D}\xspace}
\newcommand{\EdgeProb}[2]{\ensuremath{q_{#1}(#2)}\xspace}
\newcommand{\influv}[1][]{\ensuremath{\mathbf{f}_{#1}}\xspace}
\newcommand{\exav}[1][]{\ensuremath{\mathbf{q}_{#1}}\xspace}
\newcommand{\IniProb}[1]{\ensuremath{p^0_{#1}}\xspace}
\newcommand{\Vp}{\ensuremath{V^{+}}\xspace}
\newcommand{\Ut}[1][]{\ensuremath{U({#1})}\xspace}

\newcommand{\ExpAct}[1][]{\ensuremath{\ifthenelse{\equal{#1}{}}{f}{f({#1})}}\xspace}
\newcommand{\SocWel}[1][]{\ensuremath{\ifthenelse{\equal{#1}{}}{W}{W({#1})}}\xspace}
\newcommand{\esw}[1][]{\ensuremath{%
\ifthenelse{\equal{#1}{}}{\overline{W}}{\overline{W}({#1})}}\xspace}
\newcommand{\TOTALINF}{\ensuremath{\phi}\xspace}
\newcommand{\totalinf}[2][]{\ensuremath{%
\ifthenelse{\equal{#1}{}}{\TOTALINF(#2)}{\TOTALINF_{#1}({#2})}}\xspace}

\newcommand{\pr}[1][]{\ensuremath{\mbox{Pr}[{#1}]}\xspace}
\newcommand{\SVEC}[1][]{\ensuremath{\ifthenelse{\equal{#1}{}}{\mathbf{y}}{\mathbf{y}^{(#1)}}}\xspace}
\newcommand{\SVECD}[1]{\ensuremath{Y_{#1}}\xspace}
\newcommand{\svec}[2][]{\ensuremath{%
\ifthenelse{\equal{#1}{}}{y_{#2}}{y^{(#1)}_{#2}}}\xspace}
\newcommand{\uvec}{\ensuremath{\mathbf{\hat{u}}}\xspace}

\newcommand{\demand}{\emph{Demand}\xspace}
\newcommand{\onehop}{\emph{One-Hop}\xspace}
\newcommand{\nonet}{\emph{No-Network}\xspace}
\newcommand{\dummy}{Empty\xspace}
\newcommand{\latency}[2][]{\ensuremath{\Delta_{{#1}\rightarrow{#2}}}\xspace}
\newcommand{\threshold}{\ensuremath{\Gamma}\xspace}
\newcommand{\degree}[1]{\ensuremath{d_{#1}}\xspace}
\newcommand{\harvard}[1][]{Harvard\ensuremath{_{\threshold=#1\mathrm{ms}}}\xspace}

\renewcommand{\GenEquation}[4][]{%
\begin{eqnarray*}
#2 & #3 & #4
\end{eqnarray*}}

\newcommand{\dkcomment}[1]{{\color{green}\hspace{0ex} #1}}

\begin{document}

\title{You Share, I Share: Network Effects and Economic Incentives in
  P2P File-Sharing Systems}

\date{}

\author{%
Mahyar Salek\thanks{Department of Computer Science, University of
  Southern California, CA 90089-0781, USA. E-mail: salek@usc.edu}
\and Shahin Shayandeh\thanks{
Microsoft Corporation, One Microsoft Way, Redmond, WA, 98052-6399, USA.
E-mail: shahins@microsoft.com.
Work done while the author was at the University of Southern
California.}
\and David Kempe\thanks{Department of Computer Science, University of
  Southern California, CA 90089-0781, USA. E-mail: \dkempeemailalternate. \dkempethanks}
}

\maketitle

\begin{abstract}

We study the interaction between network effects and
external incentives on file sharing behavior in Peer-to-Peer (P2P)
networks. Many current or envisioned P2P networks reward individuals
for sharing files, via financial incentives or social recognition.
Peers weigh this reward against the cost of sharing incurred when others download the shared file. As a result, if
other nearby nodes share files as well, the cost to an individual
node decreases. Such positive network sharing effects can be expected to
increase the rate of peers who share files.

In this paper, we formulate a natural model for the network effects of
sharing behavior, which we term the ``demand model.''  
We prove that the model has desirable diminishing returns properties, 
meaning that the network benefit of increasing payments decreases
when the payments are already high. 
This result holds quite generally, for submodular
objective functions on the part of the network operator.

In fact, we show a stronger result: the demand model leads to a
``coverage process,'' meaning that there is a distribution over graphs
such that reachability under this distribution exactly captures the
joint distribution of nodes which end up sharing. 
The existence of such distributions has advantages in simulating and
estimating the performance of the system. 
We establish this result via a general theorem characterizing which types of
models lead to coverage processes, and also show that all coverage
processes possess the desirable submodular properties.
We complement our theoretical results with experiments on
several real-world P2P topologies. We compare our model quantitatively
against more na\"{\i}ve models ignoring network effects. 
A main outcome of the experiments is that a good incentive scheme
should make the reward dependent on a node's degree in the network.

\end{abstract}

\section{Introduction} \label{sec:introduction}
Peer-to-Peer (P2P) file sharing systems have become an important
platform for the dissemination of files, music, and other content.
The basic idea is very simple: individuals make files available for
download from their own machine. Other users can search for files they
desire and download them from a peer who has made the file available.
Naturally, designing systems such that the search and download of
files are efficient poses many research challenges, which have
received a lot of attention in the literature
\cite{aperjis:freedman:johari:distribution,saroiu:gummadi:gribble:measurement}.

A second, and somewhat orthogonal, issue is how to ensure sufficient
participation and sharing of files. Unless enough content is provided
by individuals, the utility of membership will be very small.
If free-riding \cite{feldman:papadimitriou:chuang:stoica} is too
prevalent, the system may exhibit a quick decrease in membership
common to public-goods type economic settings \cite{schelling}.

Thus, the P2P system must be designed with incentives in mind to
encourage file sharing. These incentives can take the form of
monetary payments or redeemable ``points'' \cite{golle:leyton-brown:mironov:lillibridge},
download privileges, or simply recognition. From the system designer's
perspective, these payments should be ``small,'' while ensuring enough
participation.

On the other hand, from a peer's perspective, the payments need to be
weighed against the cost incurred by sharing a file.
In this paper, we assume that the content is shared legally and the
system is designed with security in mind: hence, the main cost to an
individual is the upload bandwidth which will be used whenever another
peer downloads a file from this node.

Nodes will in general choose to download from nearby peers (in
terms of bandwidth or latency).
Therefore, as additional nearby peers share the same files,
the load will get distributed among more nodes, and the cost to each
individual node will decrease.
Thus, not only will we expect cascading effects of sharing based on
social dynamics \cite{granovetter:threshold-models}, but
we would also expect these cascading effects to be based on a network
structure determined by point-to-point latencies and bandwidths.

Our contribution in this paper is the definition and analysis (both
theoretical and experimental) of a natural model for peers' sharing
behavior in P2P systems, in the presence of network
effects and economic incentives.
In our model, we focus only on sharing one file; in practice, the
model can be applied separately for each file of interest.
The basic premise of the model is that each node has a certain
\todef{demand} for the file. Furthermore, the network determines which
percentage of the demand will be met by downloading from each
peer sharing the file\footnote{In practice, we could expect these percentages
to correlate strongly with network latency or available bandwidth, but
our model is agnostic about the derivation.}.
The crucial implication of this model is that the more nearby peers
are sharing a file, the more evenly the demand will be distributed
among them.

The upload bandwidth cost is compensated by a \todef{payment} to the peers
who make the file available. Again, our model is agnostic about
whether these payments are monetary, recognition, or take other forms.
In our model, the payments can be explicitly based on the network
degree of peers, since high-degree nodes presumably serve a key role
in propagating sharing behavior.

We argue that this model captures the essential dynamics of
P2P systems in which a peer can join the network and download files
without sharing; hence, availability of files is not the only
incentive for sharing. The FastTrack P2P protocol, used by KaZaA,
Grokster, and iMesh, is an example where this assumption holds; hence,
our model should be a reasonable approximation for these services in
terms of its incentives.

The network operator is interested in maximizing a social welfare
function \SocWel, which grows monotonically as a function of the set
of nodes that share the file. This function could be the total number of
sharing nodes, the number of nodes with at least one uploading
neighbor, or the total download bandwidth available to peers under
various natural models of downloading.

After defining this model formally (in Section \ref{sec:preliminaries}),
we prove strong and general \emph{diminishing returns} properties
about it (in Section \ref{sec:model-analysis}).
In particular, we show that whenever \SocWel is monotone
and submodular, the network's social welfare as a function of the
payments offered to the peers is monotone i.e., increasing payments will always increase social
welfare. However the \emph{rate} of increase decreases when payments are already
high. We call the latter property diminishing returns. 

To prove this result, we consider a slightly different
model, wherein payments are combined with giving the network operator
the ability to ``force'' some set $S$ of peers to share.
By first proving certain local submodularity properties for this modified
model, the desired diminishing returns properties are implied by the general
result of Mossel and Roch \cite{mossel:roch:submodular}.
However, we derive a similar result to \cite{mossel:roch:submodular}
for a broad subclass of submodular functions which we call coverage
functions. It consists of the functions for which in the underlying
process, the distribution of nodes sharing the file is equivalent to the
distribution of nodes reachable from $S$ in an appropriately defined
random graph model. We establish this equivalence via a general and
non-trivial theorem characterizing all functions that can be
obtained by counting reachable nodes under random graph models.
As a corollary, our approach provides a much simpler proof of the main
result from \cite{mossel:roch:submodular} for coverage processes.
Moreover, the fact that the propagation of sharing behavior is a
coverage process is useful for the purpose of simulating the process
and estimating the parameters of the system, allowing more efficient
algorithms for simulations. Finally, our characterization can be of
independent interest in the study  of submodular set functions.

%

While the bulk of our paper focuses on a theoretical analysis of the
demand model, we complement the theoretical results by  an
experimental evaluation of our model (in Section
\ref{sec:experiments}),
using two network topologies derived from real-world data sets
\cite{gummadi:saroiu:gribble:king,mitking,havardking}, and a regular
two-dimensional grid topology.
We first show that network effects are significant by comparing our
demand model with one in which peers are not aware of changes in load
due to nearby sharing peers.
We then evaluate different payment schemes, in particular regarding
their dependence on nodes' degrees. We evaluate these both in terms of
the fraction of peers that end up sharing, and the amount paid by the
network operator per sharing node.

\subsection{Related Work}
There is a large body of work on incentive mechanisms in P2P
file-sharing systems. (See \cite{feldman:chuang:overcoming} for a thorough
overview and \cite{zhao:lui:chiu} for a recent generalized analysis
framework.) Incentive mechanisms can be classified in three categories:
barter-based mechanisms, reputation-based mechanisms, and
currency-based mechanisms.

\todef{Barter-based} methods \cite{anagnostakis:greenwald:exchange}
enforce repeated transactions among peers by matching each
peer to only a small subset of the network, hence raising the survival
chance for strategies based on reciprocation. This method only works when we
have a small and popular set of files. For instance, the BitTorrent
protocol~\cite{cohen:bittorrent} is a popular P2P file-sharing protocol using
this method.

\todef{Reputation-based mechanisms} have an excellent
track record at facilitating cooperation in very diverse settings, from
evolutionary biology to marketplaces like eBay. These systems keep a
tally of the contribution of each peer; the past contributions
determine which peers obtain more of the system's resources in the
future. However, the availability of cheap
pseudonyms in P2P systems makes reputation systems vulnerable to Sybil
and whitewashing attacks~\cite{feldman:papadimitriou:chuang:stoica}, leading
to ongoing work on designing sybilproof reputation
mechanisms~\cite{cheng:friedman:sybilproof}. Moreover, reputation systems may
be vulnerable to coordinated gaming strategies due to distributed
rating systems \cite{shneidman:parkes:redundancy}. 

Inspired by markets, a P2P system can also deploy a currency scheme to
facilitate resource contributions by rational peers. Generally, peers
earn currency by contributing resources to the system, and spend
the currency to obtain resources from the system. Karma~\cite{vishnumurthy:chandrakumar:sirer:karma}
is one example of this kind. \todef{Currency-based systems} may also suffer
from Sybil and whitewashing attacks, depending on their policies toward
newcomers. If newcomers are endowed with a positive balance, then the
system is vulnerable to these attacks; otherwise, there might not be enough
incentive for newcomers to join the network.
Balance control could also be troublesome, as the system might need to
deal with negative balances.

Lai et al.~\cite{lai:feldman:stoica:chuang} introduced the concept of
``private'' history vs.~``shared'' history as a way to combine
barter-based and reputation-based mechanisms in the context of an
evolutionary prisoner's dilemma. Shared history is a pool that records
peers' past behavior and services them according to their reputation.
In \cite{feldman:papadimitriou:chuang:stoica}, file sharing is modeled
as a social phenomenon, akin to those discussed by Schelling
\cite{schelling}. Users consider whether or not to contribute
files based on the number of other users who contribute.
Our model is different in that it explicitly models the costs incurred
by contributing nodes, rather than simply positing an intrinsic
generosity parameter for each user.

\section{Models and Preliminaries} \label{sec:preliminaries}
We consider a peer-to-peer network with $n$ servers (or \todef{nodes}
or \todef{peers}), and
focus on the behavior of sharing one particular file.
Thus, each peer $v$ may either choose to share the file or to not
share it. We also call sharing peers
\todef{active}, and the other ones \todef{inactive}.
The set of all peers who share is denoted by \Vp.

\subsection{The Demand Model}
Each peer has a local \todef{demand} \de[v] for the file: this
demand will originate from individual users on the server $v$ (who
themselves might not possess the file or be in a position to make it
available). The demand \de[v] should be served by downloading the
file from other servers $u \in \Vp$. The \todef{quality} of the
connection between $v$ and $u$ is captured by a matrix \PM: the larger
\p[v]{u}, the larger a fraction of $v$'s demand will be served by
$u$ (assuming that $u$ shares the file). Specifically, the demand that
$u \in \Vp$ will see from $v$ is
$\de[v] \cdot \frac{\p[v]{u}}{\sum_{w \in \Vp} \p[v]{w}}$.
The matrix \PM will in practice depend on network latencies or
bandwidth, as well as explicit download agreements. It need not be
symmetric. For the purpose of the general model, we are agnostic to
the derivation of \PM; in Section \ref{sec:experiments}, we will
derive \PM from measured network latencies by positing a latency
threshold which individuals are willing to tolerate.

A node $u \in \Vp$ sharing the file will incur a \todef{cost} of
\cost[u] per unit of demand that it serves; this cost is the result of
using upload bandwidth, machine processing time, or similar resources.
To encourage peers to share the file despite this cost, the P2P
network administrator offers payments \maxpay[u] to the nodes
$u \in \Vp$. These payments need not be the same for all nodes, and
can be derived from the network structure, e.g., a node's degree.

Different nodes may have different (and unknown) tradeoffs between
money and upload bandwidth. We model this fact by assuming that each
node $u$ has a tradeoff factor \paycoef[u], drawn independently and
uniformly at random from $[0, 1]$, which captures how many units of
bandwidth one unit of money is worth to the node.
Thus, the \todef{sharing utility} of an active node $u \in \Vp$ is
\Equation{\Ut[u]}{%
\pay[u] - \cost[u] \sum_{v} \frac{\de[v] \p[v]{u}}{%
   \sum_{w \in \Vp} \p[v]{w}},}
while the sharing utility of non-sharing nodes is 0.
(A non-sharing node does not get paid and incurs no upload costs.)
We assume that agents are rational, and thus choose whether to share
or not to share so as to maximize their own utility.

\subsection{Other Models}
As we discussed in Section \ref{sec:introduction}, one of our main
contributions is the observation that file sharing behavior should be
subject to positive network externalities, i.e., that the presence of
other sharing peers makes sharing less costly. To quantify the size
of such network effects, we define two alternative models with no or
limited effects; we will compare these two models experimentally with
the demand model in Section \ref{sec:experiments}.

\begin{enumerate}
\item In the {\nonet Model}, the peers completely ignore other sharing
  peers. Thus, a node $u$ assumes that if it shares the file, then it
  will see a fraction \p[v]{u} of the demand originating with node
  $u$. Hence, the perceived utility of node $u$ when sharing is
\Equation{\Ut[u]}{%
\pay[u]  - \cost[u] \cdot \sum_{v} \de[v] \p[v]{u}.}

\item In the {\onehop Model}, the peers are aware of network effects
  in a very limited way: node $u$ assumes that any node $v$ sharing the
  file will contribute toward serving both $v$'s and $u$'s demand,
  but not toward serving the demand of any other node $w
    \neq u,v$. Thus, the perceived utility of node $u \in \Vp$ is
  in the \onehop Model is
\Equation{\Ut[u]}{%
\pay[u] - \cost[u] \cdot \frac{\de[u] \p[u]{u}}{\sum_{w \in \Vp} \p[u]{w}}%
- \cost[u] \cdot \sum_{v \neq u} \frac{\de[v] \p[v]{u}}{%
   \sum_{w \in \Vp \cap \SET{u, v}} \p[v]{w}}.}
\end{enumerate}

\subsection{Payment Schemes, Sharing Process, and Administrator's Objective}
\label{sec:payment}
The network administrator's choice is how to set the payment offers
\maxpay[u]. In doing so, the administrator balances two competing goals:
low overall payments and high utility for the participants in the system.
In this paper, we study the impact of payment schemes on these objectives.

In order to provide enough incentives for sharing, the network
administrator should always ensure that
$\maxpay[u] \geq C_u := \cost[u] \cdot \sum_{v} \de[v]$.
Otherwise, even a node $u$ with $\paycoef[u] = 1$ (i.e., the highest
possible utility for money) would have no incentive to share the file
if no other peers are sharing the file.

The full model is thus as follows: after the administrator decides on
the payments \maxpay[u] for all nodes $u$, the random tradeoffs $\paycoef[u]$
between money and bandwidth are determined independently for all nodes $u$.
Subsequently, the process proceeds in iterations. 
In each iteration, all peers simultaneously decide whether to share
the file or not, based on the payments, costs, and previous decisions
of all other peers. The process continues until an equilibrium is
reached. Notice that because the cost to a peer is
monotone decreasing in the set \Vp of currently sharing peers, the set
of sharing peers can only become larger from iteration to iteration. 
In particular, this implies that the process will eventually terminate
with some set \Vp of active peers. 
We call this the \todef{sharing process} or \todef{activation process}.

The network administrator is in general interested in increasing
access to the file while keeping the payments low. This general
objective may be captured using various metrics.
In general, we allow for any overall social welfare function
\SocWel which increases monotonically in the set $S$ of sharing nodes.
Notice that since the set $S$ itself is the result of a random
process, the administrator's goal will be to maximize
$\Expect{\SocWel[S]}$, where $S$ is derived from the random activation
process in the demand model.
Several social welfare functions \SocWel suggest themselves naturally:
\begin{enumerate}
\item The number of active peers is a natural measure of
  participation. It is the measure frequently studied in the context
  of the diffusion of innovations or behaviors in social networks
\cite{goldenberg:libai:muller:talk,granovetter:threshold-models,InfluenceSpread,InfluenceSpreadICALP,morris:contagion,mossel:roch:submodular}.
While the objective is similar, the precise dynamics are different
between those models and the demand model.
\item The total number of \todef{serviced} nodes, i.e., nodes $v$ with
  at least one active node $u$ with $\p[u]{v} > 0$. This model is
  appropriate if we only care about how \emph{many} peers can download
  the file, but not about the quality of the connection. It implicitly
  assumes that each peer has a constant utility of 1 for downloading.
\item Each node $u$ gets a utility of $\sum_{v \in \Vp} \p[u]{v}$, and
  the social welfare is the sum of all these utilities.
  This model is based on the assumption that $u$'s demand is
  served by all of its neighbors (including possibly $u$)
  simultaneously, and that $u$'s utility is the total ``download
  bandwidth'' available in this sense.
  We call this the \emph{sum-welfare} function.
\item Each node $u$ gets a utility of $\max_{v \in \Vp} \p[u]{v}$, and
  the social welfare is the sum of all these utilities.
  This is based on the assumption that $u$'s demand is served by its
  active neighbor with the best connection, corresponding to a
  situation where parallel download from multiple sources is not
  possible.
  We call this the \emph{max-welfare} function.
\end{enumerate}

Notice that the social welfare function \SocWel may also include the
utilities of the sharing nodes.

\section{Theoretical Analysis of the Model} \label{sec:model-analysis}
The main analytical contribution of this paper is based on 
\todef{coverage processes}\footnote{We thank Bobby Kleinberg for this naming
  suggestion, and also note here that Theorem \ref{thm:coverage} was
  derived independently by him.}, defined formally in Definition
\ref{def:coverage}. Informally, a coverage process is a random process
such that the distribution over sets of ultimately active nodes is
also the distribution of reachable nodes under a suitably chosen
distribution of random graphs.
Our results on coverage processes are twofold:

(1) We give a general characterization of coverage processes, and show
that the activation process for P2P systems is a coverage process.
(2) We give a significantly simplified proof (compared to the general
result of \cite{mossel:roch:submodular}) showing that under coverage processes,
the expected social welfare as a function of the payments
has diminishing returns in the sense of Definition
  \ref{def:dimin}
so long as the social welfare is a submodular function of the active
nodes. 

Recall that a function $f$ defined on sets is
\todef{submodular}
if $f(S + v) - f(S) \geq f(T + v) - f(T)$ whenever $S \subseteq T$,
i.e., if the addition of an element to a larger set causes a smaller
increase in the function value than to a smaller set.
Thus, submodularity is the discrete analogue of concavity, and
intuitively corresponds to ``diminishing returns.''
An easy inductive proof (on the size of $X$) shows that submodularity
is equivalent to the condition that for all sets $X$,
\begin{equation}
f(S \cup X) - f(S) \; \geq \; (T \cup X) - f(T)
\qquad \mbox{ whenever } S \subseteq T. \label{eq:submod-equiv}
\end{equation}

\begin{definition} \label{def:dimin}
A function $g: \mathbb{R}^n \rightarrow \mathbb{R}$ has 
\todef{diminishing returns} if for every pair $i,j$ and all vectors
\xv, it satifises
\LEquation{\frac{\partial g(x_1, x_2, \dots, x_n)}{\partial x_i \partial x_j}}{0.}
\end{definition}

\begin{remark}
The notion of ``diminishing returns'' is strictly weaker than
concavity; it corresponds to concavity only along positive
coordinates axes.\footnote{We thank Shaddin Dughmi for pointing out
this ambiguity in an earlier version of the paper.} 
\end{remark}

The two main contributions of our paper together imply the following
theorem as a corollary:
\begin{theorem} \label{thm:demand-concave}
Let $\esw[{\maxpay[1], \ldots, \maxpay[n]}] = \Expect{\SocWel[S]}$
be the expected social welfare when set $S$ is obtained from the
sharing process of the demand model with payments
$\maxpay[1], \ldots, \maxpay[n]$.

If $\SocWel[S]$ is submodular, then
$\esw[{\maxpay[1], \ldots, \maxpay[n]}]$
is monotone and has diminishing returns with respect to the payments
$\maxpay[1], \ldots, \maxpay[n]$.
\end{theorem}

For the social welfare function, the diminishing returns property
intuitively means that the additional benefit in 
social welfare that can be derived from increasing the payment to a
peer $u$ decreases as the peers' current payments increase.

The proof of Theorem \ref{thm:demand-concave} is based on analyzing
the following \todef{Seed Set Model}, which we define mainly for the
purpose of analysis.
\begin{definition}[Seed Set Model]
For each node, the payment offered is $\maxpay[u] = C_u$.
Besides payments, we have a \todef{seed set} $S$ of peers
that will always share regardless of the payments.
Subsequently, the process unfolds exactly according to the sharing process.
\end{definition}

The main technical step is to show that the Seed Set Model is a
\todef{coverage process}, in the
following sense. 
\begin{definition}[Coverage Process] \label{def:coverage}
Let \totalinf{S} be the random variable describing the set of nodes
active at the end of a process starting from the set $S$ of nodes
active. The process is called a \todef{coverage process} if there
exists a distribution \DISTR over graphs $G$ such that
for each set $T$ of nodes, $\Prob{\totalinf{S} = T}$ equals 
the probability that exactly $T$ is reachable starting from $S$
in $G$ if $G$ is drawn from the distribution \DISTR.
\end{definition}

\begin{remark}

Without using our nomenclature, \cite{InfluenceSpread} showed
submodularity for the Cascade and Threshold models of innovation diffusion 
\cite{granovetter:threshold-models,goldenberg:libai:muller:talk} by
establishing that both gave rise to coverage processes.
Subsequently, \cite{InfluenceSpreadICALP} showed that there are
natural diffusion processes which are not coverage processes, yet
have a submodular function $\Expect{\SetCard{\totalinf{S}}}$.
\end{remark}

We prove that the Seed Set Model is a coverage process in two steps.
First, in Section \ref{sec:coverage-characterization}, we give a
general and complete characterization of Coverage Processes.
This characterization may be of interest in its own right, as coverage
processes have a practical advantage: they can be simulated easily and
efficiently, by first generating a random graph according to \DISTR,
and then simply finding the set of reachable nodes.

Then, in Section \ref{sec:seed-set-coverage}, we show that the Seed
Set Process satisfies the conditions established in Section
\ref{sec:coverage-characterization}.
Finally, in Section \ref{sec:concavity}, we give a simple proof that
for \emph{any} coverage process and any submodular social welfare
function, the expected social welfare under the process is also
submodular. This implies diminishing returns with respect
to the payments.

\begin{remark}
The fact that the tradeoffs \paycoef[u] between money and bandwidth
are uniformly random in $[0,1]$ is important to ensure the
submodularity and diminishing returns properties.
If the \paycoef[u] are not random but fixed, then the diminishing returns and
submodularity properties cease to hold. Furthermore, in the Seed Set
Model, the optimization problem of finding the best seed set $S$ of at
most $k$ nodes becomes very hard, as we show in the appendix.
\end{remark}

\subsection{Characterization of Coverage Processes}
\label{sec:coverage-characterization}
In this section, we characterize \emph{exactly} which random processes are
coverage processes. This theorem may be of interest in its own right,
when analyzing different processes.

Our setting is exactly as in the paper
by Mossel and Roch\cite{mossel:roch:submodular}:
each node $u$ has an activation function \LOCALACT{u}, which is monotone
non-decreasing and satisfies $\localact{u}{\emptyset} = 0$. Each node
independently chooses a threshold $\theta_u \in [0,1]$ uniformly at
random, and becomes active when $\localact{u}{S} \geq \theta_u$, where $S$ is
the previously active set of nodes.
This process is repeated until no more changes occur.

In order to express our results concisely, we use the following
discrete equivalent of a derivative
(see, e.g., \cite{vondrak:submodular-welfare}).
For a function $f$ defined on sets, we define inductively:
\begin{eqnarray*}
\der{f}{\emptyset}{S} & = & f(S)\\
\der{f}{R \cup \SET{v}}{S} & = & \der{f}{R}{S \cup \SET{v}} - \der{f}{R}{S}.
\end{eqnarray*}
It is not difficult to verify that this notion is well-defined, i.e.,
independent of \emph{which} element $v$ is chosen at which stage.

\begin{theorem} \label{thm:coverage}
The following conditions are necessary and sufficient for the process
to be a coverage process.
\begin{itemize}
\item For all sets $T$ of odd cardinality $\SetCard{T}$, as well as
  for $T = \emptyset$, and each node $u$, we have
  $\der{\LOCALACT{u}}{T}{\Compl{T}} \geq 0$.
\item For all sets $T$ of positive even cardinality $\SetCard{T}$, and each
  node $u$, we have $\der{\LOCALACT{u}}{T}{\Compl{T}} \leq 0$.
\item $\localact{u}{\emptyset} = 0$ for all $u$.
\end{itemize}
\end{theorem}

To prove this theorem, we begin with the following reasoning.
Focus on one node $u$, and its activation function \LOCALACT{u}.
If there were an equivalent graph distribution \DISTR,
then it would have to define a probability
\EdgeProb{u}{T} for the presence of edges from exactly the vertex set
$T$ to $u$. These probabilities need to satisfy the
following property: if a set $S$ of nodes is active, then the
probability of $u$ having \emph{at least} one incoming edge from $S$
must equal \localact{u}{S}.
Thus, a necessary and sufficient condition for being a coverage
function is that for each node $u$, there exists a distribution
\EdgeProb{u}{T} over sets $T$ such that
\begin{eqnarray}
\localact{u}{S} & = & \sum_{T: T \cap S \neq \emptyset} \EdgeProb{u}{T}.
\label{eqn:distribution}
\end{eqnarray}

We can express this requirement more compactly using matrix notation.
Let \influv[u] be the $(2^n - 1)$-dimensional vector consisting of all
entries of \localact{u}{S} for $S \neq \emptyset$. Similarly, let
\exav[u] be the $(2^n-1)$-dimensional vector of all \EdgeProb{u}{S}
for $S \neq \emptyset$.
Let $A$ be the $((2^n-1) \times (2^n-1))$-dimensional matrix indexed
by non-empty subsets such that $A_{S,T} = 1$ if and only if
$S \cap T \neq \emptyset$, and $A_{S,T} = 0$ otherwise.
($A$ is called an \todef{incidence matrix} \cite{brualdi:ryser}.)
Then, Equation \ref{eqn:distribution} can be rewritten as the
requirement that for each node $u$, there exists a distribution
\exav[u] such that $A \cdot \exav[u] = \influv[u]$.

For the analysis, we fix a canonical ordering of subsets. 
Specifically, if the current (sub-)universe
consists of $k$ nodes indexed $\SET{1, 2, \ldots, k}$, their canonical
ordering is defined recursively as first containing all subsets of
$\SET{1, 2, \ldots, k-1}$ in canonical order, then
the set $\SET{k}$, followed by the sets $T \cup \SET{k}$, where the
sets $T \subseteq \SET{1, 2, \ldots, k-1}$ appear in canonical order.

In order to find out when the distribution \exav[u] exists, we want to
solve the equation $A \cdot \exav[u] = \influv[u]$, or
$\influv[u] = \Inverse{A} \cdot \exav[u]$. While the inverses of some
incidence matrices have been studied before (see, e.g.,
\cite{bapat:moore-penrose}), we are not aware of any source explicitly
giving the inverse of the matrix $A$. Hence, we establish here:

\begin{lemma} \label{lem:inverse}
The inverse of $A$ is the matrix $B$ defined by
\CDFunction{b_{S,T}}{0}{\mbox{ if } S \cup T \neq \SET{1, \ldots, n}}{%
(-1)^{\SetCard{S \cap T} + 1}}
\end{lemma}

\begin{proof}
The key insight is that under the canonical ordering of sets defined
above, the matrices $A$ and $B$ can be defined recursively via
matrices $A_k$ and $B_k$.
Specifically, let $A_1 = 1$, and
\[ \begin{array}{lcl}
A_{k+1} & = &
\left( \begin{tabular}{c|c|c}

\multirow{4}{*}{$A_k$} & 0 &  \\

&$\vdots$& $A_k$\\

&0&\\
&0&\\ \hline
0 \dots 0 0 & 1 & 1 \dots 1 1 \\ \hline
\multirow{4}{*}{$A_k$} & 1 &  \\

&$\vdots$& $1$\\

&1&\\
&1&\\
\end{tabular} \right).
\end{array} \]

Similarly, let $B_1 = 1$, and
\[ \begin{array}{lcl}
B_{k+1} & = &
\left( \begin{tabular}{c|c|c}

\multirow{4}{*}{0} & 0 &  \\

&$\vdots$& $B_k$\\

&0&\\
&-1&\\ \hline
0 \dots 0 -1 & 0 & 0 \dots 0 1 \\ \hline
\multirow{4}{*}{$B_k$} & 0 &  \\

&$\vdots$& $-B_k$\\

&0&\\
&1&\\
\end{tabular} \right).
\end{array} \]
The fact that $A=A_n$ and $B=B_n$ can be observed directly from the
definition and the canonical ordering.

To prove the lemma, we can show by induction on $k$ that
$A_k \cdot B_k = I_k$ for all $k$, where $I_k$ is the $k \times k$
identity matrix. The base case $k=1$ is obvious.
For the inductive step to $k+1$, consider the $(i,j)$ entry 
$(A_{k+1} \cdot B_{k+1})_{i,j}$. 
We distinguish 7 different cases, based on the $(i,j)$ indices.
(We use $0$ to denote the $(2^{k}-1) \cdot (2^k-1)$ matrix of all
zeroes, $\mathbf{1}$ for the vector of all ones,
and $\uvec$ for the $(2^k-1)$-dimensional unit vector with 1 in its
last coordinate and 0 everywhere else.)
\begin{enumerate}
\item If $i, j < 2^k$, then the entry is 
$(A_k \cdot 0)_{i,j} + 0 + (A_k \cdot B_k)_{i,j} = (I_k)_{i,j}$ by induction hypothesis.
\item If $i > 2^k, j < 2^k$, then (writing $i'=i-2^k$), the entry is
$(A_k \cdot 0)_{i',j} - \uvec_{i'} + (\mathbf{1} \cdot B_k)_{i'} = 0$ using Lemma
\ref{lem:aux-identity}(a) below.
\item If $i < 2^k, j > 2^k$, then (writing $j'=j-2^k$), the entry is
$(A_k \cdot B_k)_{i,j'} + 0_{i,j'} - (A_k \cdot B_k)_{i,j'} = 0$.
\item If $i, j > 2^k$, then (writing $i'=i-2^k, j'=j-2^k$), the entry is
$(A_k \cdot B_k)_{i',j'} + \uvec_{j'} - ({\bf{1}} \cdot B_k)_{j'} = (I_k)_{i',j'}$, again using
Lemma \ref{lem:aux-identity}(a).
\item If $i=j=2^k$, a straightforward calculation shows that the entry is 1.
\item If $i=2^k, j < 2^k$, then the entry is 
$({\bf{1}} \cdot B_k)_j - \uvec_{j} = 0$ by Lemma \ref{lem:aux-identity}(a).
Similarly, for $i=2^k, j > 2^k$, writing $j'=j-2^k$, the entry is 
$\uvec_{j'} - ({\bf{1}} \cdot B_k)_{j'} = 0$ by Lemma \ref{lem:aux-identity}(a).
\item Finally, for $j=2^k, i < 2^k$, the entry is
$-(A_k \cdot \Transpose{\uvec})_i+(A_k \cdot \Transpose{\uvec})_i = 0$, whereas for
$j=2^k, i > 2^k$, writing $i'=i-2^k$, the entry is 
$-(A_k \cdot \Transpose{\uvec})_{i'}+1 = 0$ by Lemma \ref{lem:aux-identity}(b).
\end{enumerate}

This proves that $A_{k+1} \cdot B_{k+1} = I_{k+1}$.
\end{proof}

\begin{lemma} \label{lem:aux-identity}
Let $\mathbf{1}$ be the vector of all 1's, and $\uvec$ defined as in the
proof of Lemma \ref{lem:inverse}. Then,
(a) $\mathbf{1} \cdot B_k = \uvec$, and
(b) $A_k \cdot \Transpose{\uvec} = \mathbf{1}$.
\end{lemma}
\begin{proof}
For part (a), we show that the row sums of all rows of $B_k$ are zero
except the last row, which has a row sum of one. 
The proof is by induction. The base case $B_1 = 1$ is clear. 
For the inductive step from $k$ to $k+1$, first notice that all the 
entries in columns $j < 2^k - 1$ are zero by induction hypothesis. 
For column $2^k - 1$, the row sum of $B_k$ contributes 1 by induction
hypothesis, from which 1 is subtracted because of the entry in the
middle column.
Column $2^k$ adds up to 0 explicitly, and columns 
$j=2^k+1, \ldots, 2^{k+1}-2$ have terms of $B_k$ and $-B_k$ canceling
out. Finally, for the last column, the entries of $B_k$ and $-B_k$
cancel out, leaving the entry 1 from the middle column.

For part (b), simply notice that using part (a) and the induction
hypothesis of Lemma \ref{lem:inverse} (for $k$), we get that
$A_k \cdot \Transpose{\uvec} = A_k \cdot \Transpose{B_k} \cdot
\mathbf{1} = I_k \cdot \mathbf{1} = \mathbf{1}$.
Here, we used that $B_k$ is symmetric.
\end{proof}

The next lemma shows that so long as all \EdgeProb{u}{S} are
non-negative, by setting \EdgeProb{u}{\emptyset} appropriately, we can
always obtain a probability distribution.
\begin{lemma} \label{lem:atmostone}
With $\EdgeProb{u}{S}$ defined as $\exav[u] = B \cdot \influv[u]$,
we have $\sum_{S} \EdgeProb{u}{S} \leq 1$.
\end{lemma}
\begin{proof}
Let $\mathbf{1}$ denote the all-ones vector as before. We can rewrite
\[ 
\sum_{S} \EdgeProb{u}{S}
\; = \; \mathbf{1} \cdot (B \cdot \influv[u])
\; = \; (\mathbf{1} \cdot B) \cdot \influv[u].
\]
Using Lemma \ref{lem:aux-identity}(a), the sum is exactly equal to
$\localact{u}{\SET{1, \ldots, n}} \leq 1$, completing the proof.
\end{proof}

By Lemma \ref{lem:inverse}, we know that
$\exav[u] = B \cdot \influv[u]$. And by Lemma \ref{lem:atmostone}, the
entries sum up to at most 1. Thus, it remains to show that
the entries of \exav[u] are non-negative
if and only if \LOCALACT{u} satisfies the conditions of Theorem
\ref{thm:coverage}. To relate these formulations, we prove the
following non-recursive characterization of discrete derivatives.

\begin{lemma} \label{lem:discrete-equivalent}
For all sets $T$, we have that
\Equation{\der{f}{T}{W}}{%
\sum_{S \subseteq T} (-1)^{\SetCard{T} - \SetCard{S}} f(W \cup S).}
\end{lemma}

\begin{proof}
The proof is by induction on $\SetCard{T}$.
For $T = \emptyset$, the claim is trivial.
Now, consider a set $T_{k+1} = T_k \cup \SET{t}$ of size $k+1$.
By definition of the discrete derivative and induction hypothesis,
\begin{eqnarray*}
\der{f}{T_{k+1}}{W}
& = & \der{f}{T_k}{W \cup \SET{t}} - \der{f}{T_k}{W}\\
 &=&  \sum_{S \subseteq T_k} (-1)^{k - \SetCard{S}} f(W \cup S \cup \SET{t})
 - \sum_{S \subseteq T_k} (-1)^{k - \SetCard{S}} f(W \cup S) \\
& = & \sum_{S \subseteq T_{k+1}: t \in S} (-1)^{k+1 - \SetCard{S}} f(W \cup S)
 + \sum_{S \subseteq T_{k+1}: t \notin S} (-1)^{k+1 - \SetCard{S}} f(W \cup S)\\
& = & \sum_{S \subseteq T_{k+1}} (-1)^{k+1 - \SetCard{S}} f(W \cup S),
\end{eqnarray*}
which completes the inductive proof.
\end{proof}

\begin{extraproof}{Theorem \ref{thm:coverage}}
Fix any node $u$, and define $\exav[u] = B \cdot \influv[u]$.
By Lemma \ref{lem:discrete-equivalent}, we can write the discrete
derivative of \LOCALACT{u} at $\Compl{T}$ as
\Equation{\der{\LOCALACT{u}}{T}{\Compl{T}}}{%
\sum_{S \subseteq T} (-1)^{\SetCard{T} - \SetCard{S}}
\localact{u}{\Compl{T} \cup S}.}
Now, if $\SetCard{T}$ is odd, then
$(-1)^{\SetCard{T} - \SetCard{S}} = (-1)^{\SetCard{S} + 1}$,
so we can rewrite the above as
\[ 
\sum_{S \subseteq T} (-1)^{\SetCard{S}+1} \localact{u}{\Compl{T} \cup S}
\; = \; \sum_{W \supseteq \Compl{T}} (-1)^{\SetCard{W \cap T}+1} \localact{u}{W}
\; = \; \EdgeProb{u}{T}.
\]
Similarly, if $\SetCard{T}$ is even, then
$(-1)^{\SetCard{T} - \SetCard{S}} = (-1)^{\SetCard{S}}$, so we can
rewrite the discrete derivative as
\[ 
\sum_{S \subseteq T} (-1)^{\SetCard{S}} \localact{u}{\Compl{T} \cup S}
\; = \; \sum_{W \supseteq \Compl{T}} (-1)^{\SetCard{W \cap T}} \localact{u}{W}
\; = \; - \EdgeProb{u}{T}.
\]
Thus, the \EdgeProb{u}{T} are all non-negative (and the probability
distribution thus well-defined) if and only if
$\der{\LOCALACT{u}}{T}{\Compl{T}} \geq 0$ for $\SetCard{T}$ odd,
and
$\der{\LOCALACT{u}}{T}{\Compl{T}} \leq 0$ for $\SetCard{T} > 0$ even.
\end{extraproof}

\Omit{By claim \ref{inverse} and using the formulation of \ref{bstu}
  we have $$\exa{W}{U} = \sum_{T \supseteq U \setminus W} (-1)^{|T
    \cap W| + 1}\influ{T}{U}$$
Moreover, by lemma \ref{add2one}, we have $\sum_{S \in U} \exa{S}{U} =
1$.
}

\subsection{Coverage Property of the Seed Set Process}
\label{sec:seed-set-coverage}
In this section, we establish the following theorem.

\begin{theorem}\label{thm:p2p-coverage}
The Seed Set Process is a coverage process.
\end{theorem}

\begin{proof}
In order to prove this theorem, we want to apply Theorem
\ref{thm:coverage}.
To do so, we need to show that the local decisions of nodes about
sharing can be cast in terms of submodular threshold functions. 
Specifically, we define
\Definition{\localact{u}{S}}{%
1 - \frac{1}{C_u} \cdot \cost[u] \cdot
\sum_{v} \frac{\de[v] \p[v]{u}}{\sum_{w \in S \cup \SET{u}} \p[v]{w}}}
and let $\theta_u = 1 - \frac{\pay[u]}{C_u}$.
(Recall from Section \ref{sec:payment} that $C_u = \cost[u] \cdot \sum_{v} \de[v]$.)

A node $u$ becomes active if doing so has positive utility, i.e., if
$\pay[u] > \cost[u] \cdot
\sum_{v} \frac{\de[v] \p[v]{u}}{\sum_{w \in S \cup \SET{u}} \p[v]{w}}$.
Dividing both sides by $C_u$, and subtracting from 1 shows that this
is equivalent to saying that
\GenEquation{1 - \frac{\pay[u]}{C_u}}{<}{1 - \frac{1}{C_u} \cdot
\cost[u] \cdot \sum_{v} \frac{\de[v] \p[v]{u}}{\sum_{w \in S \cup \SET{u}} \p[v]{w}}.}
Since \pay[u] is uniformly random in $[0,C_u]$ by the definition of 
\maxpay[u] in the Seed Set Model, this condition is equivalent to saying that
$\theta_u < \localact{u}{S}$. Thus, we have shown that the activation process
can be equivalently recast in terms of threshold activations
functions.

Finally, we need to show that for every node $u$, all derivatives
\der{\LOCALACT{u}}{T}{S} are non-negative when $\SetCard{T}$ is odd
and non-positive when $\SetCard{T} > 0$ is even.
(The fact that $\localact{u}{S} = \der{\LOCALACT{u}}{\emptyset}{S}$ is
non-negative follows directly by definition.)
Let
\Equation{\localc{u}{x_1, \dots, x_n}}{%
1 - \frac{1}{C_u} \cdot \cost[u] \cdot
\sum_{v} \frac{\de[v] \p[v]{u}}{\sum_{v_i \in V} \p[v]{v_i} x_i}}
be the continuous equivalent of the local influence function
\LOCALACT{u}.
For a set $S$, let \SVEC[S] denote the $n$-dimensional vector with
$\svec[S]{i} = 1$ if $v_i \in S \cup \SET{u}$ and
$\svec[S]{i} = 0$ otherwise.
Then, $\localact{u}{S} = \localc{u}{\SVEC[S]}$.
Notice that by definition, there is no division by zero.

Writing $d \SVECD{T} = d\svec{i_1} d\svec{i_2} \dkcomment{\cdots} d\svec{i_{\SetCard{T}}}$,
where $T = \SET{i_1,i_2, \ldots, i_{\SetCard{T}}}$,
an easy inductive proof first shows that
\Equation{\der{\LOCALACT{u}}{T}{S}}{%
\int_0^1 \dots^{|T|} \int_0^1
\frac{d \localc{u}{\SVEC[S]}}{d \SVECD{T}} d \SVECD{T}.}

It remains to show that each term inside the integration is
non-negative for odd $|T|$ and non-positive for even $|T|$.
We accomplish this by showing that
\Equation{\frac{d \localc{u}{\SVEC[S]}}{d \SVECD{T}}}{%
(-1)^{|T| + 1} |T|! \frac{\cost[u]}{C_u}
\sum_{v} \frac{\de[v] \p[v]{u} \prod_{t \in T} \p[v]{t}}{%
(\sum_{v_i \in V} \p[v]{v_i} \svec[S]{i})^{|T|+1}}.}

The proof is by induction.
The base case: $|T| = 1$ can be verified easily.
Assume that the claim holds for $|T| = i-1$. We have
\begin{eqnarray*}
\frac{d \localc{u}{\SVEC[S]}}{d \SVECD{T} d \svec{i}}
&=&
\frac{d}{d \svec{i}} (-1)^{|T| + 1} |T|!
\frac{\cost[u]}{C_u} \sum_{v}
\frac{\de[v] \p[v]{u} \prod_{t \in T} \p[v]{t}}{%
(\sum_{v_i \in V} \p[v]{v_i} \svec{i})^{|T|+1}}\\
&=&
(-1)(-1)^{|T| + 1} |T|! \frac{\cost[u]}{C_u} \cdot
\sum_{v} \frac{(|T| + 1) \p[v]{v_i} \de[v] \p[v]{u}
                  \prod_{t \in T} \p[v]{t} (\sum_{v_i \in V} \p[v]{v_i} \svec{i})^{|T|}}{%
                 (\sum_{v_i \in V} \p[v]{v_i} \svec{i})^{2|T|+2}}\\
&=&
(-1)^{|T| + 2} |T+1|! \frac{\cost[u]}{C_u}
\sum_{v} \frac{\de[v] \p[v]{u} \prod_{t \in T \cup \SET{v_i}} \p[v]{t}}{%
(\sum_{v_i \in V } \p[v]{v_i} \svec{i})^{|T|+2}}.
\end{eqnarray*}
This completes the inductive proof, and thus the proof of
Theorem \ref{thm:p2p-coverage}.
\end{proof}

While we defined the Seed Set Process primarily as a tool for
analysis, we remark here that Theorem \ref{thm:p2p-coverage} has a
direct consequence for the optimization problem of
maximizing the expected total number of active nodes at the end of the
process, subject to a size constraint on the seed set $S$.
A Theorem of Nemhauser et
al.~\cite{cornuejols:fisher:nemhauser,nemhauser:wolsey:fisher}
states that if $f$ is any non-negative, monotone, and submodular
function on sets, then the greedy algorithm is a polynomial-time
$(1-1/e)$-approximation (where $e$ is the base of the natural
logarithm). Since we can approximate the expected number of active
nodes under the Seed Set Process arbitrarily closely by simulating the
activation process (see \cite{InfluenceSpread} for an in-depth
discussion of the greedy algorithm), we obtain the following corollary:

\begin{corollary} \label{cor:approximation}
The best starting set $S$ for the Seed Set Process
can be approximated within $(1 - 1/e - \epsilon)$ in polynomial time,
for any $\epsilon > 0$.
\end{corollary}

\subsection{Diminishing Returns of Expected Social Welfare}
\label{sec:concavity}

Finally, we use the machinery of coverage processes to show
diminishing returns of social welfare.
Consider an arbitrary coverage process.
When the coverage process starts with the set $T$, let \totalinf{T} be
a random variable describing the set of nodes active at the end of the
process. Thus, the distribution of \totalinf{T} for all $T$ precisely
characterizes the coverage process. Our main theorem is now the
following:

\begin{theorem} \label{thm:covsub}
Let $h(S)$ be any monotone submodular function of $S$.
Then, $\Expect{h(\totalinf{T})}$ is a monotone submodular function of
$T$, where the expectation is taken over the randomness in \totalinf{T}.
\end{theorem}
This theorem follows from the general result of
\cite{mossel:roch:submodular}, since all coverage processes are
locally submodular, and our utility function is submodular with
respect to the set of sharing neighbors.
However, below we give a
very simple proof based on reachability in graphs using the fact that
\TOTALINF is a coverage process. This is useful for the purpose of
simulating the process and estimating \TOTALINF. It means that instead
of generating random thresholds and simulating a dynamic process, we
can generate a random graph and then simply use BFS
to find the number of reachable nodes.

\begin{proof}
Because \TOTALINF is a coverage process, by Theorem
\ref{thm:coverage}, there is a distribution \pr[\cdot]
over graphs $H$ such that for any set $T$, the set of nodes reachable
in $H$ from $T$ has the same distribution as \totalinf{T}.
Let \totalinf[H]{T} denote the set of nodes reachable from
$T$ in $H$. Then,
\Equation{\Expect{h(\totalinf{T})}}{%
\sum_{H} \pr[H] \cdot h(\totalinf[H]{T}).}

Fix some graph $H$ and let $S \subseteq T$ and $x \notin T$.
Then,
\begin{eqnarray*}
h(\totalinf[H]{T + x}) - h(\totalinf[H]{T})
& = &
h(\totalinf[H]{T} \cup \totalinf[H]{\SET{x}}) - h(\totalinf[H]{T})\\
& \leq &
h(\totalinf[H]{S} \cup \totalinf[H]{\SET{x}}) - h(\totalinf[H]{S})\\
& = &
h(\totalinf[H]{S + x}) - h(\totalinf[H]{S}),
\end{eqnarray*}
where the inequality followed from Inequality (\ref{eq:submod-equiv}),
and the equalities from the definitions of reachability in a graph.
Thus, for any fixed graph $H$, the function
$h(\totalinf[H]{T})$ is monotone and submodular in $T$.
Because the \pr[H] are probabilities,
$\Expect{h(\totalinf{T})}$ is a non-negative linear combination of
monotone submodular functions, and thus also monotone and submodular.
\end{proof}

The final piece of the proof of Theorem \ref{thm:demand-concave} is
the following lemma, showing that monotonicity and submodularity of
the Seed Set Model imply diminishing returns for the
original model. 
\begin{lemma} \label{lem:submodular-concave}
Let $f$ be a non-negative, monotone, submodular function on sets.
Consider the function $g$ defined as follows:
Each element $u$ is included in $S$ independently with
probability \InclProb{u}{\maxpay[u]}, where \INCLPROB[u] is an
increasing and concave function of \maxpay[u].
Define $g(\maxpayv) = \Expect{f(S)}$.
Then, $g$ is monotone and satisfies the diminishing returns property
as defined in Definition \ref{def:dimin}.
\end{lemma}

\begin{proof}
First, notice that 
$g(\maxpayv) = \sum_{S' \subseteq \V} f(S') \prod_{u \in S'}
\InclProb{u}{\maxpay[u]} \prod_{u \notin S'} (1 - \InclProb{u}{\maxpay[u]})$. 
In order to show the diminishing returns property, it is enough to
show that
$\frac{\partial g(\maxpayv)}{\partial \maxpay[i]} \geq 0$ and
$\frac{\partial g(\maxpayv)}{\partial \maxpay[i] \partial \maxpay[j]} \leq 0$ 
for all $i, j \in \V$. Using the definition of $g$, we have: 
\begin{eqnarray*}
\label{first-deriv}
\frac{\partial g(\maxpayv)}{\partial \maxpay[i]} 
&=& \phantom{-} \sum_{S \subseteq V, i \in S} f(S) 
  \cdot \frac{d \InclProb{i}{\maxpay[i]}}{d \maxpay[i]}  
  \cdot \prod_{u\in S, u \neq i} \InclProb{u}{\maxpay[u]}
  \cdot \prod_{u \notin S}(1-\InclProb{u}{\maxpay[u]})\\
&& - \sum_{S \subseteq V, i \notin S} f(S) 
  \frac{d \InclProb{i}{\maxpay[i]}}{d \maxpay[i]}   
  \cdot \prod_{u\in S} \InclProb{u}{\maxpay[u]}
  \cdot \prod_{u \notin S, u \neq i}(1-\InclProb{u}{\maxpay[u]})\\ 
&=& \sum_{S \subseteq V, i \in S} (f(S) - f(S - i) )
  \cdot \frac{d \InclProb{i}{\maxpay[i]}}{d \maxpay[i]}
  \cdot \prod_{u\in S, u \neq i} \InclProb{u}{\maxpay[u]}
  \cdot \prod_{u \notin S}(1-\InclProb{u}{\maxpay[u]}) \\
& \geq & 0.
\end{eqnarray*}

The last inequality holds because 
$\frac{d \InclProb{i}{\maxpay[i]}}{d \maxpay[i]} \geq 0$ and $f$ is monotone. 

Next we need to show that 
$\frac{\partial g(\maxpayv)}{\partial \maxpay[i] \partial \maxpay[j]} \leq 0$ 
for all $i, j \in \V$. For $i = j$, a calculation similar to the one
above shows that
\begin{eqnarray*}
\frac{\partial^2 g(\maxpayv)}{\partial \maxpay[i]^2} 
& = & 
\sum_{S \subseteq V, i \in S} (f(S) - f(S - i))
  \cdot \frac{d^2 \InclProb{i}{\maxpay[i]}}{d \maxpay[i]^2}  
  \cdot \prod_{u\in S, u \neq i} \InclProb{u}{\maxpay[u]}
  \cdot \prod_{u \notin S}(1-\InclProb{u}{\maxpay[u]}),
\end{eqnarray*}
which is non-positive because $f$ is monotone and $q_i$ is concave.

Finally, suppose that $i \neq j$. 
Using a calculation similar to the one above, we can rewrite 
$\frac{\partial g(\maxpayv)}{\partial \maxpay[i] \partial \maxpay[j]}$
as 
\[ \sum_{S \subseteq \V \setminus \{i,j\}} 
\left(f(S+i+j) - f(S+i) -f(S+j) + f(S)\right)
\cdot \frac{d \InclProb{i}{\maxpay[i]}}{d \maxpay[i]} 
\cdot \frac{d \InclProb{j}{\maxpay[j]}}{d \maxpay[j]} 
\cdot \prod_{u\in S} \InclProb{u}{\maxpay[u]} 
\cdot \prod_{u \notin S, u \neq i,j}(1-\InclProb{u}{\maxpay[u]}),
\]
which is non-positive because $f$ is submodular and $q_i, q_j$ are
concave.
\end{proof}

With Theorem \ref{thm:p2p-coverage} and Lemma
\ref{lem:submodular-concave}, we can now complete the proof of Theorem
\ref{thm:demand-concave}.

\begin{extraproof}{Theorem \ref{thm:demand-concave}}
Consider one node $u$. The probability that it becomes active
initially is
\[ 
\IniProb{u} \; = \; \Prob{\pay[u] \geq C_u} \; = \; 1-\frac{C_u}{\maxpay[u]}.
\]
Recall that $C_u = \cost[u] \cdot \sum_{v} \de[v]$, and $\maxpay[u]
\geq C_u$ in our model, so this number is always non-negative.

Clearly, \IniProb{u} is also a monotone increasing function of \maxpay[u].
To verify concavity, we simply take two derivatives: the second
derivative is $\frac{-2C_u}{(\maxpay[u])^3}$, and thus non-positive, so
\IniProb{u} is concave.

Now, consider all the nodes $u$ which did not initially become active.
This is equivalent to saying that $\pay[u] \leq C_u$. But subject to
this bound, \pay[u] is uniformly random, so we are in the situation of
having an initially active set $S$, and for each remaining node $u$,
the payment is independently and uniformly random in $[0,C_u]$.
By Theorems \ref{thm:p2p-coverage} and \ref{thm:covsub},
the expected social welfare \esw[S] is a monotone and submodular
function of the seed set $S$, so long as $\SocWel$ is submodular in
the set of active nodes.
We can therefore apply Lemma \ref{lem:submodular-concave} to
$\Expect{h(\totalinf{T})}$, which implies that
$\esw[{\maxpay[1], \ldots, \maxpay[n]}]$ has the diminishing returns property.
\end{extraproof}

Each of the social welfare functions
listed in Section \ref{sec:preliminaries} can be shown to be monotone
and submodular in the set of active nodes by simple calculations.
Thus, for all of these objective functions, the total social welfare
is a monotone function of the payments with diminishing returns properties.

\section{Experimental evaluation} \label{sec:experiments}
In this section, we summarize our observations based on
simulations both on synthetic and real-world P2P networks.

We have developed a simulator for the three models described in
Section~\ref{sec:preliminaries}.
\subsection{Simulation model}

Given a payment scheme \maxpayv, we generate random \paycoef[u] and compute
the number of active (sharing) nodes. We also compute the value of the
social welfare according to the utility functions
in Section \ref{sec:preliminaries}.

In addition, we calculate the total payments, and the average
payment per active and per serviced node. These numbers are averaged
over 1000 iterations, each with different random \paycoef.

\noindent {\bf Network topology}. For our evaluation, we consider
different network topologies, including two network topologies derived
from real-world data sets
\cite{gummadi:saroiu:gribble:king,mitking,havardking}, and a regular
two-dimensional grid topology.
The real-world data sets are based on measured end-to-end latencies
between pairs of servers deployed in the Internet
\cite{gummadi:saroiu:gribble:king}.
The MIT King data set \cite{mitking} is symmetric and measures RTT
between each pair among 1740 servers, while the Harvard King data set
\cite{havardking} provides asymmetric median latencies between each
pair among 1895 servers. In addition to networks derived from these two data sets, we also
consider a regular two-dimensional grid.

We derive the download percentage matrix \PM from the latencies
by setting $\p[v]{u} = \max(0, 1 - \frac{\latency[u]{v}}{\threshold})$,
where \latency[u]{v} is the latency from $u$ to $v$,
and \threshold is a hard threshold for tolerable latencies.
This models the fact that users prefer to download from peers to which
they have fast connections, and have a threshold beyond which latency
may not be tolerable any more. By varying \threshold, we can obtain
denser or sparser download network topologies.
We will refer to the networks derived from the MIT King data set as
\todef{MIT networks}, and those derived from the Harvard King data set
as \todef{Harvard networks}.

In addition to networks derived from these two data sets, we also
consider a regular two-dimensional grid. We
do not report all results for all topologies here. Unless stated
otherwise, our observed trends apply to all of these topologies.

\noindent {\bf Payment schemes and non-sharing peers}.
In our experiments, we consider different payment schemes \maxpayv, to
study the impact of payments on the propagation of sharing behavior.
We parameterize the schemes with two parameters $\alpha, \beta$, and
set $\maxpay[u] = \alpha \cdot \degree{u}^\beta$, where
\degree{u} is the degree of node $u$ in the network defined by the
\p[u]{v} values. Thus, the financial utilities are chosen uniformly at
random from the interval $[0, \alpha \cdot \degree{u}^\beta]$.

We also consider the impact of peers who cannot (or do not want to)
share the file at all, regardless of the payment offered. Such peers
may still be interested in downloading the file. Their presence can be
expected to decrease the sharing behavior in networks, as they will
place load on other peers without contributing.
We call such nodes ``\dummy'' nodes, and consider the impact of different
percentages of \dummy nodes on the overall sharing percentage.

\begin{figure*}
\begin{center}
\begin{tabular}{cc}
  \psfig{figure=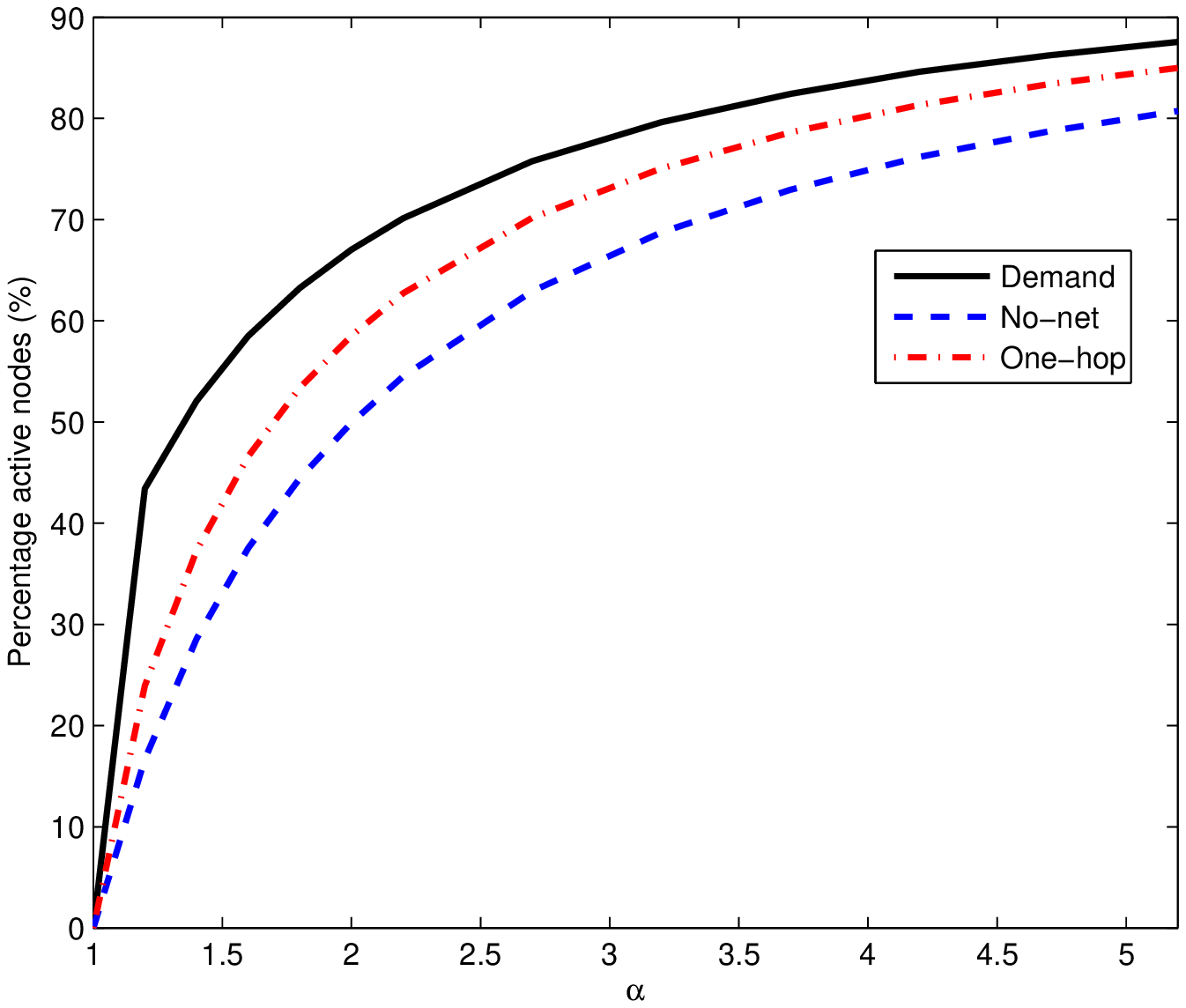,width=3.0in} &
    \psfig{figure=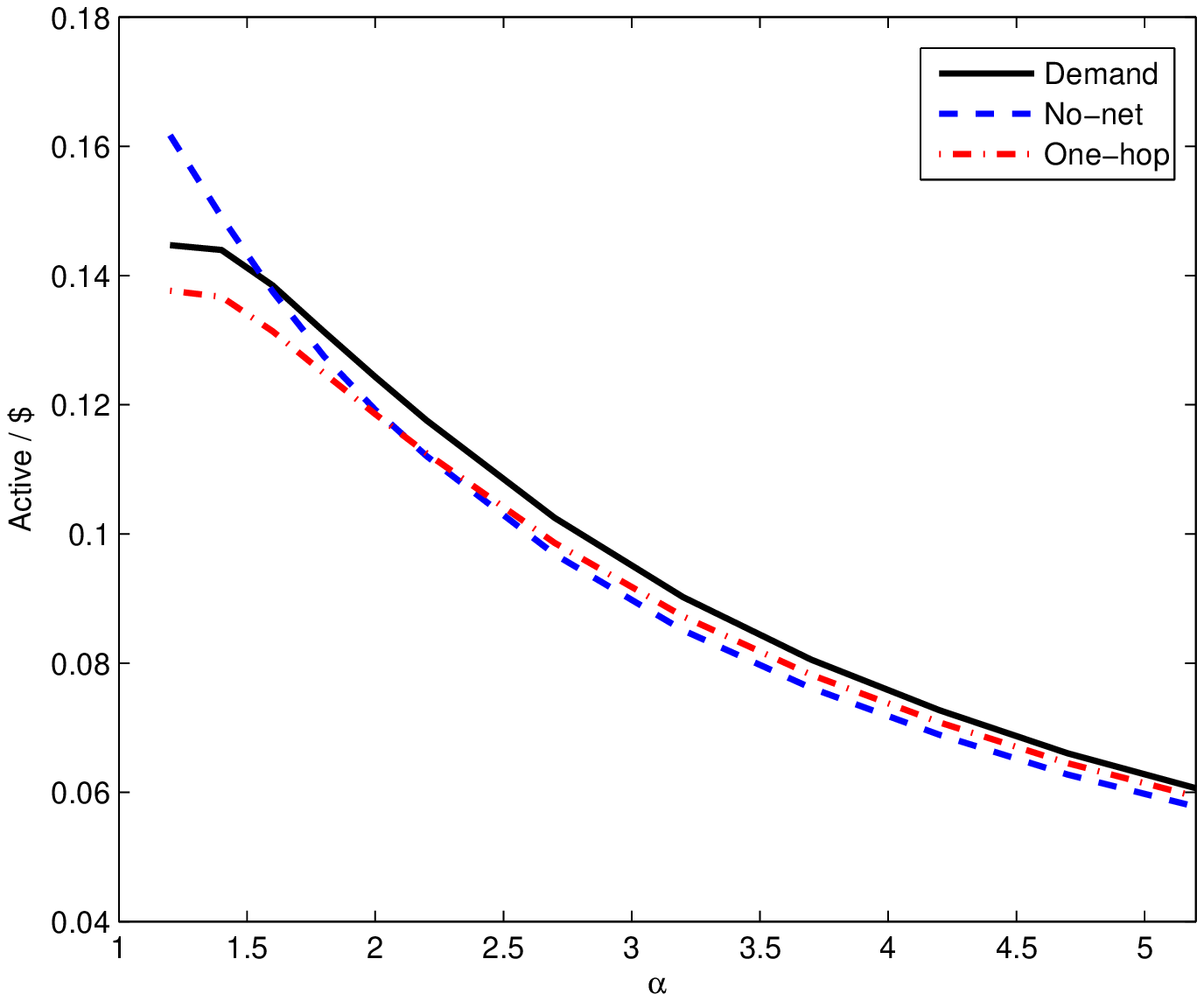,width=3.0in} \\
    (a) Percentage active nodes &
    (b) Active nodes per unit of payment \\
\end{tabular}
\caption{Comparison of different models, using the \todef{Harvard network} with
  no \dummy nodes.}\label{fig:diffmodel}
\end{center}
\end{figure*}

\subsection{Results}
\noindent {\bf Comparison of different models}. We begin by estimating
the size of network effects, by comparing the \demand model with the
\nonet and \onehop models.
Figure~\ref{fig:diffmodel} (a) compares the participation rates under
the three models, with the same payment scheme and same network
(Harvard). We keep $\beta = 1$ constant in the payment scheme, and
vary $\alpha$. Thus, payments are proportional to nodes' degrees.
The figure shows that by ignoring network effects, we would
underestimate the number of sharing nodes by about 15\% on average,
and as much as 25\% (for $\alpha=1.2$). The same trends hold for the
fraction of \emph{serviced} nodes (not shown here): the number of
serviced nodes is underestimated by about 10\% if ignoring network
effects.

Figure~\ref{fig:diffmodel} (b) compares the number of active nodes per
unit of payment spent by the network administrator. This is an
interesting metric as it captures the tradeoff between participation
and payments. Compared to the number of active nodes, the choice of
model seems to have remarkably little impact on the estimate of this
quantity. For small values of $\alpha$, the network effects lead to
slightly higher payments per active nodes, as the network effects lead
to an activation of more high-degree nodes, which have higher
payments. This effect disappears as $\alpha$ increases, and more nodes
are activated in the \nonet model as well. The same trends hold for the number of
\emph{serviced} nodes per unit of payment spent (not shown here).

The results reported here stay essentially the same both for the MIT
and grid topologies. In particular, the underestimate of the number of
active nodes by the \nonet model is essentially the same in these
topologies. In the grid topology, the \nonet model in fact
overestimates the cost per active node by about 10\%, as the
dependence on the degree disappears, and network effects lead to an
activation of more nodes with smaller payments.

\begin{figure}
 \begin{center}
 \begin{tabular}{c}
   \psfig{figure=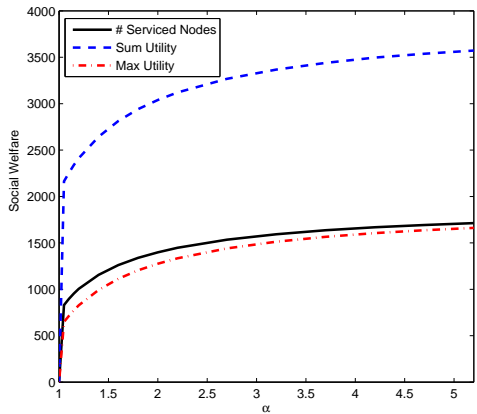,width=2.8in}
\end{tabular}
   \caption{Sum-Welfare, Max-Welfare and the number of serviced nodes, \todef{Harvard network}.}\label{fig:diffutil}
 \end{center}
\end{figure}

\noindent {\bf Different Social Welfare Functions}
We evaluate our theoretical results against the number of serviced nodes
and the two social welfare functions sum-welfare and max-welfare, as
defined in Section \ref{sec:preliminaries}.
All three are plotted in Figure~\ref{fig:diffutil}.
Although each social welfare function differs from
the others in terms of the degree of submodularity (for example,
sum-welfare can be shown to be completely modular in the number of
active nodes), the curvatures of the plots as a function of payments
are more or less the same. Thus, the concavity (diminishing returns)
appears to be dominated by the submodularity of the activation process.

\noindent {\bf Different payment schemes}.
For a network administrator, it is particularly interesting how the
choice of payments will affect sharing behavior, and the
cost-effectiveness of achieving a certain participation rate. Our next
set of experiments therefore shows the percentage of active nodes, and
the number of active nodes per unit of payment, when the parameters
$\alpha$ and $\beta$ in the payments
$\maxpay[u] = \alpha \cdot \degree{u}^\beta$ are varied.

Figure~\ref{fig:diffscheme} (a) shows the percentage of
active nodes in \todef{Harvard network} with 50\% \dummy nodes,
as a function of $\alpha$ and $\beta$.
Figure~\ref{fig:diffscheme} (b) shows the
number of active nodes per unit of payment under the same setting.
The cost effectiveness is maximized for very small values of $\alpha$
and $\beta$, specifically $\beta = 0$ and $\alpha = 1.2$. However,
this comes at a steep price, in that almost no nodes (only about 4.4\%
of the network) share in this case.

Clearly, there is no single point at which the network should
operate. Rather, a network administrator who wants to achieve a
certain participation rate can use these plots find the most
cost-effective payment scheme to achieve this rate. For instance,
if the goal is to achieve 30\% sharing, this can be achieved by
setting $\alpha=1.6$ and $\beta=1.5$, or $\alpha=1.8$ and $\beta=1$.
Of these, the first scheme spends about 30 units per active node,
while the second scheme spends about 7 units per active node. Thus, a
judicious choice of payments can lead to significant savings while
ensuring the same level of participation.

In general, the plot suggests that $\beta \in [0.5,1]$ tends to lead
to good tradeoffs between participation and cost: for smaller values
of $\beta$, participation tends to be too low, while for higher
values, the cost per active node increases significantly.

%

The observed trends are fairly independent of the network topologies.
In particular, the plots for both the grid and \todef{MIT network}
also suggest that $\beta \in [0.5,1]$ gives the best cost efficiency for
a given fraction of participating nodes.

\begin{figure*}
\begin{center}
\begin{tabular}{cc}
  \psfig{figure=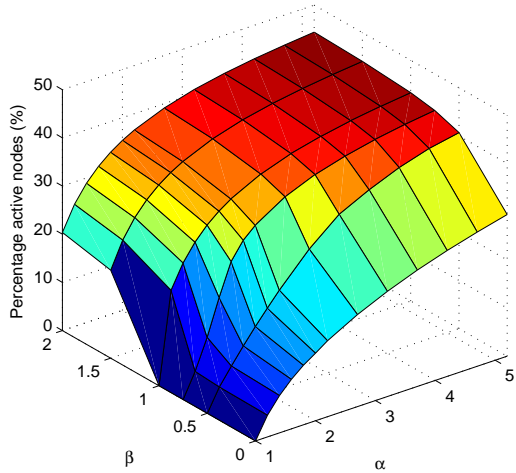,width=3.0in} &
    \psfig{figure=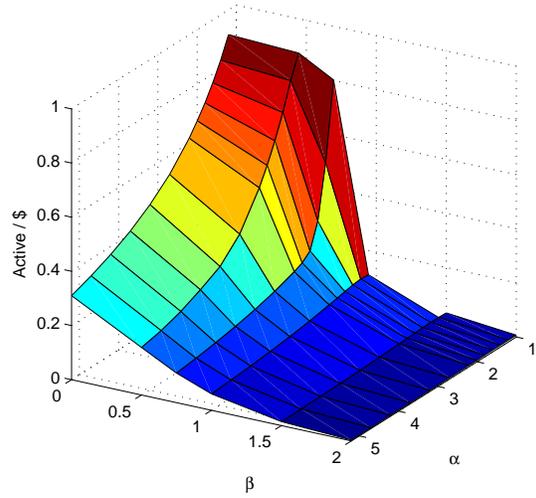,width=3.0in} \\
    (a) Percentage of nodes active &
    (b) Active nodes per unit of payment \\
\end{tabular}
  \caption{Comparison of different payment schemes, \todef{Harvard
      network} with 50\% \dummy nodes.
Notice that for readability, the directions of axis labels
  for $\alpha, \beta$ is different in the two figures.}\label{fig:diffscheme}
\end{center}
\end{figure*}

\noindent {\bf Different thresholds (\threshold)}.
Finally, we investigate the impact of different latency tolerance
thresholds \threshold on the activation process.
Recall that the larger \threshold, the more peers $u$ may serve $v$.
For instance, with $\threshold=2\mbox{ms}$, the average degree of
nodes in the \todef{Harvard network} is 4.58, while with
$\threshold=5\mbox{ms}$, the average degree increases to 14.93.
In the resulting denser graph, we would expect less degree imbalance,
and overall higher network effects; however, the payments will need to
compensate for more downloads from any individual node.

The experiments, conducted on the \todef{Harvard network} with no \dummy nodes,
confirm this intuition.
When $\beta=0$, Figure~\ref{fig:threshold} (a) shows that
the number of nodes serviced is smaller in \harvard[5] than in
\harvard[2]. The reason is that the payments do not increase with
the degree, so it is costlier for nodes in \harvard[5] to become
active.
As $\beta$ increases, and high degrees result in higher compensation,
more nodes are serviced in \harvard[5].
With $\beta > 0$, payments increase in the node degree, and nodes in
\harvard[5] receive more payments because of their higher average
degree. Thus, more nodes are activated, and as a result, more nodes
can be serviced.

The increased activation comes at a price, as seen in Figure
\ref{fig:threshold} (b). The higher average degree in \harvard[5],
combined with the dependence of payments on the degrees, leads to
somewhat higher payments per active (or serviced) node. Thus, in the \demand
model, the increased participation in denser networks is not only a
result of network effects, but also of higher payments.

Therefore, in order to investigate the effectiveness of density itself
on the participation or service rate, we make the following
comparison. Fix the payment per active node for both \harvard[5] and
\harvard[2] to an arbitrary number by choosing the appropriate payment
schemes for each graph. For example, in order to get a payment of 23
per active node, a payment scheme for \harvard[5] would be
$\alpha=2.7$ and $\beta=1$ and for \harvard[2] would be $\alpha=1$ and
$\beta=1.5$. It turns out that the denser network (\harvard[5]) gives
a significantly higher rate for both participation and service. For
a payment of 23 units per active node, for instance, the fraction of
participating nodes for \harvard[5] is 86\% while the same fraction
goes down to 39\% in \harvard[2].

\begin{figure*}
\begin{center}
\begin{tabular}{cc}
  \psfig{figure=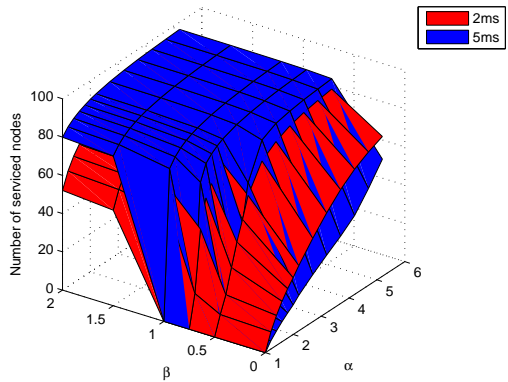,width=3.0in} &
    \psfig{figure=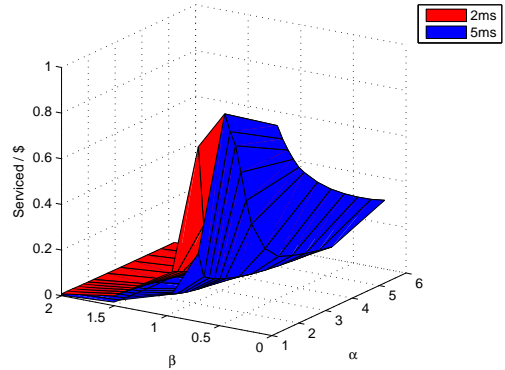,width=3.0in} \\
    (a) Fraction of serviced nodes &
    (b) Serviced nodes per unit of payment \\
\end{tabular}
  \caption{Comparison of different thresholds, \todef{Harvard network}
    with no \dummy nodes.}
\label{fig:threshold}
\end{center}
\end{figure*}

Based on the simulations, the following were our main observations:

\begin{enumerate}
\item How different are the predictions in sharing behavior between
the \demand, the \nonet, and the \onehop models?
Our results show a significant difference between the models in their
prediction of sharing: while the fraction of sharing nodes is
qualitatively similar, the predictions ignoring network effects can be
off by about 15\%--25\%. This results in up to 10\% depreciation in the number of serviced peers.
\item How does the participation depend on the network topology and
  density? 
  We observe that the denser the network,
  the higher the rate of participation, given fixed incentives.
This holds across grid and realistic Internet topologies.  
\item How does the payment scheme affect the number of sharing and
  serviced nodes, and the price paid per node?
  Our experiments suggest that the payments \maxpay[u] for realistic
  topologies should be proportional to $u$'s degree to give high
  overall participation at low cost. In other words, given a network topology, there exists a choice of parameters for payments proportional to node degrees that maximizes the overall ``bang per buck''. We derive these parameters for each network topology experimentally.
\end{enumerate}

\section{Conclusions} \label{sec:conclusions}
There are several natural directions for future work.  A very interesting question arises
when taking payments by ``reputation'' or download
priorities into account. While monetary
payments can (in principle) be increased arbitrarily, reputation is
inherently constant-sum: if some peers are recognized as outstanding
sharers, then others will receive less recognition, and might find the
reduced recognition not enough incentive to keep sharing.
Similarly, download priorities come at the expense of other
  peers, and can thus not be arbitrarily increased for all members of
  the network.
As a result, the process of sharing will not necessarily be
monotone: peers may choose to stop sharing once too many
  other peers are active. A first question is then whether stable
(equilibrium) states even exist.
If so, it would be interesting what fraction of the peers will be
sharing, what the social welfare is, and how these quantities will
depend on the network structure.

From a more practical viewpoint, it would
be desirable to evaluate how accurately our model (or a variation
thereof) captures the actual behavior of participants in a P2P
system. This would likely be a difficult experiment to perform, as
many of the parameters, such as file demands and latency, are
inherently transient, and in a realistic system, payments cannot be
changed constantly to evaluate the impact of such changes.

In the bigger picture, the network designer also has to be
concerned about manipulation by peers. For instance, colluding peers
could artificially inflate the perceived ``degree'' of a peer (by
claiming a download preference), and thus the payments to that peer.
A more thorough investigation of mechanisms taking these and other
concerns into account is an exciting direction for future work.

Finally, our work lies among various applications in economics for
which there are positive or negative externalities among agents in a
neighborhood. Our results suggest that in order to study different
economic metrics such as revenue or social welfare, we should always
consider the cascading effect of agents' strategies over the network.

\bibliographystyle{plain}
\bibliography{names,conferences,publications,bibliography}

\begin{appendix}

\section{Hardness of Approximation under the Seed Set Model}

Here, we prove that finding a seed set $S$ to (even approximately)
maximize the eventual number of active nodes is hard under the Sharing
Process.
Let \textsc{Best Seed} be the optimization problem of finding the seed set $S$ of at
most $k$ nodes that maximizes the total number of sharing nodes, given
$n$ servers $u_1, \dots, u_n$ and the corresponding parameters 
$\cost[u], \de[u], \paycoef[u], \pay[u], \p[v]{u}$.
(Notice that when all of the \paycoef[u] are given, the process is
deterministic.)

\begin{proposition} \label{prop:hardness}
It is hard to approximate \textsc{Best Seed} within $n^{1-\epsilon}$
for any $\epsilon > 0$ unless P = NP.
\end{proposition}

\begin{proof}
We reduce from the \textsc{Vertex Cover} problem. 
Recall that the \textsc{Vertex Cover} problem is formulated as
follows: Given a graph $G=(V,E)$, a vertex cover is a set
$S \subseteq V$ of nodes such that each edge $e \in E$ has at least
one endpoint in $S$. In the \textsc{Vertex Cover} decision problem,
the input is a pair $(G,k)$: the question is
whether there is a vertex cover of size at most $k$. 
We assume without loss of generality that $G$ contains no isolated
vertices.

Given an arbitrary \textsc{Vertex Cover} instance with 
$N = \SetCard{V}$ nodes and $M = \SetCard{E} \geq N/2$ edges,
we construct an instance of \textsc{Best Seed} as follows:
For each node $u \in V$, we have a node $w_u$. 
For each edge $e \in E$, we create two nodes $x_e, x'_e$.
Finally, setting $r=1/\epsilon$, we create $M^r$ ``bulk'' nodes
$y_1, \ldots, y_{M^r}$. 
We set $\p[x'_e]{y_i} = 1$ for all $y_i, x_e$. 
For all $e$, $\p[x'_e]{x_e} = 1$. 
Finally, whenever $e$ is incident on $u$, we have $\p[x_e]{w_u} = 1$.
All other values of $p$ are 0.

We visualize the construction above in 4 layers. 
The ``node layer'' consists of all nodes $w_u$ for all $u$. 
The ``primary layer'' consists of all $x_e$. 
The ``secondary layer'' consists of all $x'_e$. 
Finally, the ``bulk layer'' consists of all $y_i$.
Next, we define payments and demands:

\begin{eqnarray*}
\label{setcoverpayments}
\pay[v] = \left\{ \begin{array}{rl}
 0 &\mbox{ if $v = w_u$ for some $u \in V$ (node layer)} \\
  3.5 &\mbox{ if $v = x_{e}$ for some $e \in E$ (primary layer)}\\
  0 &\mbox{ if $v = x'_{e}$ for some $e \in E$ (secondary layer)}\\
  M + 0.5 &\mbox{ otherwise (bulk layer)}
       \end{array} \right.
\end{eqnarray*}

\begin{eqnarray*}
\label{setcoverdemands}
\de[v] = \left\{ \begin{array}{rl}
 0 &\mbox{ if $v = w_u$ for some $u \in V$ (node layer)} \\
  2 &\mbox{ if $v = x_{e}$ for some $e \in E$ (primary layer)}\\
  2 &\mbox{ if $v = x'_{e}$ for some $e \in E$ (secondary layer)}\\
  0 &\mbox{ otherwise (bulk layer)}
       \end{array} \right.
\end{eqnarray*}

First, let $T$ be a vertex cover of size at most $k$. Consider the effect
of starting with the nodes $w_u, u \in T$ as a seed set. 
Because $T$ is a vertex cover, each primary node $x_e$ now has
an active node $w_u$ with $\p[x_e]{w_u} = 1$, so that its demand of
2 is split between itself and (at least) one node $w_u$. Thus, upon
activation, it would face at most a demand of $2$ from $x'_e$ and 1
from itself, whereas its payment is $3.5$. Hence, each
primary node will become active in the second round.
Once the primary node $x_e$ is active, $x'_e$ will split its demand
evenly between $x_e$ and all active bulk nodes. Hence, each bulk node
$y_i$ will see demand at most $1$ from each $x'_e$, for a total of
$M$. Since its payment offer is larger, $y_i$ will become
active. Hence, all bulk nodes will be active by round 3, and the total
number of active nodes is at least $M^r+M+k$.

Conversely, suppose that strictly more than $M+N$
nodes are active. Because none of the secondary nodes ever become
active (since they have a payment offer of 0), this means that at
least one bulk node must be active. Let $y_i$ be the first bulk node
to become active, breaking ties arbitrarily. Because no other bulk
nodes are active at this time, $y_i$ must see demand at least $1$ from each
secondary node $x'_e$. And because its payment offer
is only $M+0.5$, this means that it cannot see demand $2$ from
any secondary node --- otherwise, the total demand would exceed the
payment. This means that for each secondary node $x'_e$, the
corresponding primary node $x_e$ must already be active. 
Without loss of generality, the seed set contained no primary nodes
--- otherwise, the node $x_e$ could be replaced by $w_u$ (where $u$ is
an endpoint of $e$), which would next activate $x_e$. 
Thus, $x_e$ must have become activated at some point of the process,
which can only happen when its total demand is smaller than its
payment. Since at that point, only $x_e$ can serve the demand of
$x'_e$, this in turn means that $x_e$'s own demand must be split between
itself and one or more active nodes $w_u$. Thus, if $S$ is the set of
initially active nodes in the node layer, then the corresponding
vertices of $G$ must form a vertex cover.

In summary, if there is a vertex cover of size at most $k$, then there is
a seed set of size at most $k$ activating at least $M^r+M+k$ nodes,
whereas otherwise, no seed set of size at most $k$ can activate more
than $M+N \leq 3M$ nodes. 
Thus, no approximation better than $\Omega(M^{r-1})$ is possible.
Since the total number of nodes is $n=M^r+2M+N \leq 2M^r$ (for $r$
large enough), 
this proves an approximation hardness of $\Omega(n^{1-1/r}) =
\Omega(n^{1-\epsilon})$, unless P=NP.
\end{proof}

\end{appendix}
\end{document}